\documentclass{amsart}

\usepackage{xcolor}

\usepackage{amsmath}
\usepackage{amsthm}
\usepackage{stmaryrd}
\usepackage{amsfonts}
\usepackage{mathtools}
\usepackage{enumitem}

\theoremstyle{plain}
\newtheorem{theo}{Theorem}
\newtheorem{lemme}{Lemma}
\newtheorem{prop}{Proposition}

\newcommand{\llb}{\llbracket}
\newcommand{\rrb}{\rrbracket}

\DeclarePairedDelimiter\floor{\lfloor}{\rfloor}
\DeclarePairedDelimiter\abs{\lvert}{\rvert}

\newcommand{\mybinom}[2]{\Bigl(\begin{array}{@{}c@{}}#1\\#2\end{array}\Bigr)}

\newcommand{\Z}{\mathbb{Z}}

\newcommand{\F}{\mathbb{F}}
\newcommand{\Fm}{\F_m}

\def\rg#1{{\color{black}#1}}

\newcommand{\cale}{{\mathcal{E}}}

\title{Periodic autocorrelation of sequences}

\author{Fran\c cois Rodier}
\address{Aix Marseille Univ, CNRS, Centrale Marseille, I2M, Marseille, France} 
\curraddr{}
\email{francois.rodier@univ-amu.fr}
\thanks{}

\author{Florian Caullery}
\address{Qualcomm Office, Sophia-Antipolis, France}
\curraddr{}
\email{fcauller@qti.qualcomm.com}
\thanks{}

\author{Eric F\'erard}
\address{\'Equipe GAATI, Universit\'e de la Polyn\'esie Fran\c{c}aise}
\curraddr{}
\email{eric.ferard@upf.pf}
\thanks{}

\begin{document}

\small

\date{}

\begin{abstract}
  The autocorrelation of a sequence is a useful criterion, among all, of resistance to cryptographic attacks.
  The behavior of the autocorrelations of random Boolean functions (studied by Florian Caullery, Eric F\'erard and Fran\c cois Rodier \cite{CR})
  shows that they are concentrated around a point.
  We show that the same is true for the evaluation of the periodic autocorrelations of random binary sequences. 
\end{abstract}



\maketitle

\section{Introduction}

In this article, we are interested in random  sequences of rational integers.
The most interesting case occurs when the entries are just \(-1\) or \(1\), in which case we call the sequence binary. More precisely we are interested \rg{in} the periodic autocorrelation of the sequences that we are going to define now.

Let \(m\) be a prime number.
Let \(\F_m = \{ 0, 1, \ldots, m-1\}\) and \(\F_m^* = \{ 1, \ldots, m-1\}\) where the elements are taken modulo \(m\).
Let \(S_m = \{ s_0, s_1, \ldots, s_{m-1} \} \in \{-1, 1\}^m\).
We endow the set \(\{ -1, 1 \}^m\) with a uniform probability distribution, so that the \(s_i\)'s are independent and
equally likely to take the value \(-1\) or \(1\).
We define
\[C_u(S_m) = \sum_{i \in \F_m} s_is_{i+u}\]
and the periodic autocorrelation of the sequence \(S_m\)
\[C(S_m) = \max_{u \in \F_m^*} \bigl| C_u(S_m) \bigr|.\]

We find an evaluation of the mean of the periodic autocorrelations of the random sequences.
It happens to be the point of accumulation of the periodic autocorrelations of random  sequences.

We therefore want to prove the following theorem.

\begin{theo}\label{theo}~{}
  \begin{enumerate}
  \item[a)] The expectation (denoted $\cale$) of the periodic autocorrelation has the following limit:
    \[\frac{ \rg{{\cale}(C(S_m))} }{ \sqrt{m\log m} } \longrightarrow \sqrt{2}\]
    as \rg{the primes} \(m \to +\infty\).
  \item[b)] As \rg{the primes} \(m\rightarrow \infty\), in probability
    \[\frac{C_m}{ \sqrt{m \log m}} \rightarrow \sqrt 2.\]
  \end{enumerate}
\end{theo}

\rg{ We assume that $m$ is prime. In fact, computer calculations seem to show that it seems not necessary but it may be much more complicated.
}

There are several works that survey topics involving correlations of sequences.
Most of them focus on particular aspects.
Jungnickel and Pott \cite{JP} and Cai and Ding \cite{CD}, concentrate on optimal binary sequences and cyclic difference sets.
Golomb and Gong \cite{GG} deal with theoretical aspects of binary sequences with nearly ideal autocorrelation functions and the applications of these sequences.
Helleseth and Kumar \cite{HK} pay particular attention to sequences with low correlation. Other works focus on aperiodic autocorrelations.
\rg{We keep that same notation. We define the aperiodic autocorrelation of $S_m$ at shift $u$ by
$$C^{ap}_u(S_m) = \sum_{0\le k, k+u<m} s_k s_{k+u}$$
The relation between $m$-periodic and aperiodic autocorrelation reads like that
$$C_u=C^{ap}_u+C^{ap}_{m-u}$$
}
We will not deal here with aperiodic autocorrelation. See the article by Schmidt \cite{KUS4} which discuss the analogous problem for aperiodic autocorrelation.

{
On the other hand
Mauduit and S\'ark\"ozy introduced and studied certain numerical parameters associated to finite binary sequences  in order to measure their ``level of randomness":
Normality measure,
Well-distribution measure,
Correlation measure.
But  they were designed for the aperiodic autocorrelation of pseudo\-random sequences whereas those we study are related
to the periodic autocorrelation of random sequences.} See Cassaigne, Mauduit and S\'ark\"ozy \cite{CMS} and Kai-Uwe Schmidt \cite{KUS3}.

The idea of this paper came from the proof that in the similar case of  all random Boolean functions, these functions  accumulate around the expected values of their nonlinearity.
It was proved by Schmidt in \cite{KUS2}, finalising the work of Rodier \cite{FR}, Dib \cite{Dib} and Litsyn and Shpunt \cite{LS},
that the nonlinearity of random Boolean functions
is concentrated around its expected value. Similarly, as there does not exist a
study of the distribution of the periodic autocorrelation of 
random  sequences of rational integers,
we fill the gap with our result by showing that the same phenomenon happens
in the case of the periodic autocorrelation.

{We follow the same scheme of proof as for the nonlinearity of random Boolean functions, except for the lower bound of the expectation of autocorrelations of random sequences which requires a more involved result. Namely, we evaluate the autocorrelation expectation of random sequences by calculating the number of even sequences with certain properties (see the \rg{Appendix})}. 

After the introduction, we state some important propositions in a  {\sl preliminary} section. Then we prove the main theorem, and finally, in an  \rg{Appendix}, we proceed to the proof of
the lower bound of the expectation of the autocorrelations of the random sequences, which is a tricky result.

\section{Preliminaries}

The \(X_{x,u} = s_x s_{x+u}\) for \(x\) in \(\F_m\) are not mutually independent.
But we can prove that  \(X_{x,u} = s_x s_{x+u}\) for \(x\) in \(\F_m^*\) are mutually independent.
For that we adapt the proof of Mercer \rg{\cite[Prop 1.1]{Mer}}.

\begin{lemme}
  Let \(u \in \Fm^*\).
  The \(X_{x,u} = s_xs_{x+u}\) for \(x \in \Fm^*\) are mutually independent.
\end{lemme}

\begin{proof}
  Since \(x \longmapsto u^{-1}x\) is an automorphism of \(\Fm\), we can assume that \(u = 1\).
  Let~\(E\) be a subset of \(\Fm^*\).
  We must prove that
  \[P \Bigl( \bigcap_{x \in E} (X_{x,1} = b_x) \Bigr) = \prod_{x \in E} P(X_{x,1} = b_x) = 2^{-\# E}\]
  where the \(b_x\) are \(\pm 1\).

  \medskip

  Let \(G\) be the graph whose vertices are the elements of \(\Fm\) and whose edges are precisely the pairs of the from \((x, x+1)\) where \(x \in E\).
  It is a subgraph of the graph with vertices the elements of \(\Fm\) and with edges the pairs \((x,x+1)\) where \(x \in \Fm^*\).
  Since \(m\) is prime, the latter is a path from \(1\) to \(0\).
  Hence, the graph \(G\) is a disjoint union of connected paths.

  \medskip

  Let \(H\) be a connected subpath of \(G\) of length lower than \(m - 1\).
  We can assume that it is a path from \(1\) to \(r\).
  Let \(b_1, \ldots, b_r \in \{ \pm 1 \}\).
  We have
  \begin{align*}
    P(X_{1,1} = b_1, \ldots, X_{r,1} = b_r) &= P(s_1s_2 = b_1, \ldots, s_rs_{r+1}=b_r) \\
                                            &= P(s_1 = 1, s_2 = b_1, \ldots, s_{r+1} = b_1 \cdots b_{r}) \\
                                            &\qquad + P(s_1 = -1, s_2 = -b_1, \ldots, s_{r+1} = -b_1 \cdots b_{r}).
  \end{align*}
  Since \(r < m\), the variables \(s_1, s_2, \ldots, s_{r+1}\) are independent.
  Hence,
  \[
    P(X_{1,1} = b_1, X_{2,1} = b_2, \ldots, X_{r,1} = b_r) = \frac{1}{2^{r}} = P(X_{1,1} = b_1) P(X_{2,1} = b_2) \cdots P(X_{r,1} = b_r).
  \]

  Let now \(H\) and \(H'\) be two disjoint connected subpaths of \(G\) of length lower than \(m - 1\).
  Since \(H \cap H' = \emptyset\), the events \(\bigcap_{x \in H} (X_{x,1} = b_x)\) and \(\bigcap_{x' \in H'} (X_{x',1} = b_{x'})\) are independent.
  So, we have
  \begin{align*}
\lefteqn
    {P\Bigl(\bigcap_{x \in H} (s_xs_{x+1} = b_x) \cap \bigcap_{x' \in H'} (s_{x'}s_{x'+1} = b_{x'}) \Bigr) }
    \\
    &= P\Bigl( \bigcap_{x \in H} (s_xs_{x+1} = b_x) \Bigr) P\Bigl( \bigcap_{x' \in H'} (s_{x'}s_{x'+1} = b_{x'}) \Bigr) \\
    &= \prod_{x \in H} P(X_{x,1}=b_x) \prod_{x' \in H'} P(X_{x',1}=b_{x'}).
  \end{align*}
  Since \(G\) is a disjoint union of connected paths, we can conclude.
\end{proof}

We derive two consequences on bounds involving \(S_m\).

\begin{prop}\label{prop_C(S_m)_div_sqrt}
  For all \(\epsilon > 0\), as \(m \to +\infty\),
  \[P \Bigl( \frac{C(S_m)}{\sqrt{2m\log m}} > 1 + \epsilon \Bigr) \longrightarrow 0.\]
\end{prop}

\begin{proof}
  The union bound gives
  \[P(C(S_m) > \mu_m) \le \sum_{u \in \Fm^*} P(\abs{C_u(S)} > \mu_m) = \sum_{u \in \Fm^*} P \Bigl( \bigl\lvert \sum_{x \in \Fm} X_{x,u} \bigr\rvert > \mu_m \Bigr)\]
  with \(\mu_m = (1 + \epsilon)\sqrt{2m\log m}\).
  Since \(\lvert \sum_{x \in \Fm} X_{x,u} \rvert \le \abs{X_{0,u}} + \lvert \sum_{x \in \Fm^*} X_{x,u} \rvert\) and \(\abs{X_{0,u}} = 1\), we have
  \[P(C(S_m) > \mu_m) \le \sum_{x \in \Fm^*} P \Bigl( \bigl\lvert \sum_{x \in \Fm^*} X_{x,u} \bigr\rvert > \mu_m - 1 \Bigr).\]
  As the variables \(X_{1,u}, \ldots, X_{m-1,u}\) are mutually independent, we can apply Corollary A.1.2 of \cite{AS} with \(k = m\) to obtain
  \[P(C(S_m) > \mu_m) \le \rg{2   e^{-\frac{(\mu_m - 1)^2}{2m-2}}},\]
  which tends to \(0\) as \(m \to +\infty\).
\end{proof}

We need this lemma from H. Cram\'er \cite {cr}.

\begin{lemme}
  Let \(X_0, X_1, \ldots\) be identically distributed mutually independent random variables satisfying \(\cale[X_0] = 0\) and \(\cale[X_0^2] = 1\)
  and suppose that there exists \(T > 0\) such that \(\cale[e^{tX_0}] < \infty\) for all \(\abs{t} < T\).
  Write \(Y_k = X_0 + X_1 + \cdots + X_{k-1}\) and let \(\Phi\) be the distribution function of a normal random variable with zero mean and unit variance.
  If \(\theta_k > 1\) and \(\theta_k/k^{1/6} \to 0\) as \(k \to \infty\), then
  \[\frac{P\bigl( \lvert Y_k \rvert \ge \theta_k\sqrt{k}\bigr)}{2\Phi(-\theta_k)} \to 1.\]
\end{lemme}

We can now apply this lemma to obtain the following proposition.

\begin{prop}\label{prop_P(abs(C_u(S_m))}
  For all \(m\) sufficiently large,
  \[P \bigl( \abs{C_u(S_m)} \ge \sqrt{2m\log(m)} \bigr) \ge \frac{1}{2m\sqrt{\log m}}.\]
\end{prop}

\begin{proof}
  Since \(\abs{X_{0,u}} = 1\), we have \(P(\abs{X_{1,u} + \cdots + X_{m-1,u}} \ge \sqrt{2m\log(m)} + 1) \le P( \abs{C_u(S_m)} \le \sqrt{2m\log(m)} )\) and it suffices to prove that
  \[P(\abs{X_{1,u} + \cdots + X_{m-1,u}} \ge \sqrt{2m\log(m)} + 1) \ge \frac{1}{2m\sqrt{\log m}}.\]
  Notice that \(\cale(e^{tX_{1,u}}) = \cosh(t)\).
  Write \(\sqrt{2m\log(m)} + 1 = \xi'_m\sqrt{m-1}\) with \(\xi'_m = \frac{1}{\sqrt{m-1}}( \sqrt{2m\log(m)} + 1)\).
  We have \(\xi'_m > 1\) and \(\lim_m \frac{\xi'_m}{m^{1/6}} = 0\).
  So we can apply the previous lemma to obtain
  \[P(\abs{X_{1,u} + \cdots + X_{m-1,u}} \ge \sqrt{2m\log(m)} + 1) \sim 2\Phi(-\xi'_m).\]
  For all \(z > 0\), we have
  \[\frac{1}{\sqrt{2\pi}z}\Bigl( 1 - \frac{1}{z^2} \Bigr)e^{-z^2/2} \le \Phi(-z) \le \frac{1}{\sqrt{2\pi}z}e^{-z^2/2}.\]
  So, as \(m \to +\infty\),
  \[2\Phi(-\xi'_m) \sim \frac{1}{m\sqrt{\pi \log(m)}},\]
  from which the proposition follows.
\end{proof}

\begin{prop}\label{prop_P(Cu_cap_Cv)}
  Write \(\lambda_m = \sqrt{2m\log m}\).
  For all \(m\) sufficiently large and for all distinct \(u, v \in \Fm^*\) such that \(u + v < \frac{m}{2\log m}\), we have
  \[ P( \abs{C_u(S_m)} \ge \lambda_m \cap \abs{C_v(S_m)} \ge \lambda_m) \le \frac{6e^2}{m^2}. \]
\end{prop}

This result is proven in the  \rg{Appendix} (see section \ref{annex}).

\section{Proof of theorem \ref{theo}}

By using an inequality from martingales theory (see McDiarmid \cite{McD}), we can find an upper bound for \(\abs{\cale(C(S_m)) - C(S_m)}\)
(see Caullery-Rodier \cite{CR}).

\begin{lemme}\label{lemme_McDiarmid}
  For all \(\theta > 0\), we have
  \[P(\abs{C'(S_m) - \cale(C'(S_m))} \ge \theta) \le 2 \exp \Bigl( -\frac{\theta^2}{8(m-1)} \Bigr)\]
  where
  \[C'(S_m) = \max_{u \in \F_m^*} \lvert \sum_{\substack{i \in \F_m^* \\ i \not= -u}} s_is_{i+u} \rvert.\]
\end{lemme}

\begin{proof}
  See \rg{section 4 of} Caullery-Rodier \cite{CR}.
\end{proof}

\begin{lemme}\label{lemme_McDiarmid2}
  For all \(\theta' > 4\), we have
  \[P\bigl( \abs{\cale(C(S_m)) - C(S_m)} \ge \theta' \bigr) \le 2\exp \Bigl(-\frac{(\theta' - 4)^2}{8(m-1)} \Bigr).\]
\end{lemme}

\begin{proof}
  Let \(\theta' > 4\) and \(\theta = \theta' - 4\).
  We check that \(\lvert C(S_m) - \cale(C(S_m)) \rvert \le 4 + \lvert C'(S_m) - \cale(C'(S_m)) \rvert\).
  So, by lemma \ref{lemme_McDiarmid}, we have
  \[P\bigl( \theta' < \abs{C(S_m) - \cale(C(S_m))} \bigr) \le 2\exp \Bigl(-\frac{\theta^2}{8(m-1)} \Bigr). \qedhere\]
\end{proof}

We can obtain a lower bound of \(C(S_m)\).

\begin{lemme}\label{lemme_min_P(C(S_m))}
  For all \(m\) sufficiently large, we have
  \[P(C(S_m) \ge \lambda_m) \ge \frac{1}{15 \log^{3/2}m}\]
  where \(\lambda_m = \sqrt{2m\log(m)}\).
\end{lemme}

\begin{proof}
  Let \(m\) be an integer greater than \(2\) and let \(W = \bigl\{ u \in \Fm : 1 \le u \le \frac{m}{4 \log m} \bigr\}\).
  For all \(m\) sufficiently large, we have
  \[\frac{m}{6\log m} \le \abs{W} \le \frac{m}{4\log m}.\]
  Then
  \begin{align*}
    P(C(S_m) \ge \lambda_m) &\ge P(\max_{u \in W} \abs{C_u(S_m)} \ge \lambda_m) \\
                            &\ge \sum_{u \in W} P(\abs{C_u(S_m)} \ge \lambda_m)
                              - \sum_{\substack{u, v \in W \\ u<v}} P(\abs{C_u(S_m)} \ge \lambda_m \cap \abs{C_v(S_m)} \ge \lambda_m)
  \end{align*}
  by the Bonferroni inequality.
  For all \(m\) sufficiently large, we have
  \[
    \sum_{u \in W} P(\abs{C_u(S_m)} \ge \lambda_m ) \ge \sum_{u \in W} \frac{1}{2m\sqrt{\log m}} = \abs{W} \cdot \frac{1}{2m\sqrt{\log m}}
    \ge \frac{1}{12 \log^{3/2} m}
  \]
  by proposition \ref{prop_P(abs(C_u(S_m))}.
  Let \(u, v \in W\) with \(u < v\).
  We have \(u + v < \frac{m}{2\log m}\).
  By proposition~\ref{prop_P(Cu_cap_Cv)}, we have
  \[P(\abs{C_u(S_m)} \ge \lambda_m \cap \abs{C_v(S_m)} \ge \lambda_m) \le \frac{6e^2}{m^2}\]
  for all \(m\) sufficiently large.
  We obtain then
  \[
    P(C(S_m) \ge \lambda_m) \ge \frac{1}{12\log^{3/2} m} - \frac{e^2}{3\log^2 m} \ge \frac{1}{15\log^{3/2} m}
  \]
  for all \(m\) sufficiently large.
\end{proof}

Finally we have

\begin{theo}
  The following limit holds when \(m \to +\infty\)
  \[\frac{\cale(C(S_m))}{\sqrt{m\log m}} \longrightarrow \sqrt{2}.\]
\end{theo}
\
\begin{proof}
  Let \(\epsilon > 0\).
  By the union bound and triangle inequality, we have
  \[
    P \Bigl( \frac{\cale(C(S_m))}{\sqrt{m\log m}} - \sqrt{2} > \epsilon \Bigr) \le
    P \Bigl( \frac{\cale(C(S_m))}{\sqrt{m\log m}} - \frac{C(S_m)}{\sqrt{m\log m}} > \frac{1}{2} \epsilon \Bigr)
    + P \Bigl( \frac{C(S_m)}{\sqrt{m\log m}} - \sqrt{2} > \frac{1}{2} \epsilon \Bigr).
  \]
  The right hand side of the last inequality goes to zero as \(m \to +\infty\)
  by proposition \ref{prop_C(S_m)_div_sqrt} and lemma \ref{lemme_McDiarmid2}.
  So we conclude
  \[\limsup_{m \to +\infty} \frac{\cale(C(S_m))}{\sqrt{m\log m}} \le \sqrt{2}.\]
  The proof of the claim is based on an idea in \cite{LS}: to bound by below \(\frac{\cale(C(S_m))}{\sqrt{m\log m}}\),
  we will prove that the following set is finite.
  Let \(\delta > 0\) and define
  \[N(\delta) = \Bigl\{ m > 1 : \frac{\cale(C(S_m))}{\sqrt{m\log m}} < \sqrt{2} - \delta \Bigr\}.\]
  For sake of contradiction, we assume that this set is infinite
  Then, for all \(m \in N(\delta)\) sufficiently large, we have \(\lambda_m - \cale(C(S_m)) > 4\) 
  %(where \(\lambda_m = \sqrt{2m\log m}\)) 
  and so
  \[\frac{1}{15 \log^{3/2}m} \le \rg{ P(C(S_m) \ge \lambda_m)} \le 2\exp \Bigl( -\frac{(\lambda_m - \cale(C(S_m)) - 4)^2}{8(m - 1)} \Bigr),\]
  by lemmas \ref{lemme_min_P(C(S_m))} and \ref{lemme_McDiarmid2}.
  \rg{Hence, for all \(m \in N(\delta)\) sufficiently large, we have
  $\lambda_m - \cale(C(S_m)) - 4 > \delta\sqrt{m\log m}-4 > \frac{\delta}{ 2}\sqrt{m\log m}$
%  \[\frac{\cale(C(S_m))}{\sqrt{m\log m}} \ge \sqrt{2} - \frac{4}{\sqrt{m\log(m)}} - \sqrt{\frac{8\log(30) + 12 \log\log m}{\log m}},\]
%  which contradict the fact that \(N(\delta)\) is infinite.
and hence
\[  \frac{1}{15 \log^{3/2}m} \le 
2\exp \Bigl( -\frac{ \delta^2 \log m}{32} \Bigr)
 = \frac{2}{m^{\delta^2 /32} }\]
 which cannot happen for $m$ sufficiently large.
 }

\end{proof}

\begin{prop}
  We have
  \[\frac {C(S_m) }{ \sqrt {2m \log m}} \rightarrow 1\]
  in probabilities.
\end{prop}

\begin{proof}
  It is enough to show that
  \(\lim_{m\rightarrow\infty} P\Big( \Bigl \lvert\frac{C(S_m)  }{  \sqrt {2m \log m}} -1\Bigr\rvert> \epsilon\Big)=0\).

  We have by the triangular inequality
  \[P\bigg( \Bigl\lvert\frac{C(S_m)  }{  \sqrt {2m \log m}} -1\Bigr\rvert > \epsilon\bigg) \le P\bigg( \Bigl\lvert\frac{\cale(C(S_m\ ))  }{  \sqrt {2m \log m}} - \frac{C(S_m)  }{  \sqrt {2m \log m}}\Bigr\rvert > \epsilon/2\bigg) + P\bigg( \Bigl\lvert \frac{\cale(C(S_m ))  }{  \sqrt {2m \log m}} -1\Bigr\rvert > \epsilon/2\bigg) .\] 
  By lemma \ref{lemme_McDiarmid} the term 
  \(P\Big(\frac{ \lvert \cale(C(S_m))-C(S_m) \rvert  }{  \sqrt {2m \log m} } \ge \epsilon\Big)\)
  tends to 0 as \(m\rightarrow\infty\).

  On the other hand, the term
  \(P\Big( \Bigl\lvert \frac{\cale(C(S_m))  }{  \sqrt {2m \log m}} -1\Bigr\rvert > \epsilon/2\Big) \)
  is zero except for a finite set as we have just seen.

  So the proposition is true.
\end{proof}

\section{ \rg{Appendix} : proof of the proposition \ref{prop_P(Cu_cap_Cv)}\label{annex}}

In this section, we will prove the proposition \ref{prop_P(Cu_cap_Cv)}.
So, we would like to find an upper bound of
\[ P \bigl( (\abs{C_u(S_m)} \ge \lambda_m) \cap (\abs{C_v(S_m)} \ge \lambda_m) \bigr). \]

Let \(p\) be a positive integer, \(a, b \in \Fm^*\) and \(\theta_1, \theta_2 > 0\).
By Markov's inequality and since 
\[ (\abs{C_u(S_m)} \ge \lambda_m) \cap (\abs{C_v(S_m)} \ge \lambda_m) \Longrightarrow
\Big( C_u(S_m) C_v(S_m) \Big)^{2p} \ge (\theta_1 \theta_2) ^{2p} \]
%===============
 we have
\begin{align*}
  P( \abs{C_a(S_m)} \ge \theta_1 \cap \abs{C_b(S_m)} \ge \theta_2 )
  &\le P \Bigl( \bigl( \sum_{i \in \Fm} s_is_{i+a} \ge \theta_1 \bigr)^{2p} \cap \bigl( \sum_{j \in \Fm} s_js_{j+b} \ge \theta_2 \bigr)^{2p} \Bigr) \\
  &\le \frac{1}{(\theta_1\theta_2)^{2p}}\cale\Bigl( \bigl( \sum_i s_is_{i+a} \sum_j s_js_{j+b} \bigr)^{2p} \Bigr).
\end{align*}
Before going on, we need some definitions.
Let \(n\) be a positive integer.
A sequence \((u_1, \ldots, u_{2n})\) of  \(\Fm\) is called {\bfseries even} if for every \(\lambda \in \Fm\)
the set of the \(u_i\)'s equals to \(\lambda\) has even cardinal (see Schmidt \cite{KUS1}).
Let \(\xi = (a_1, \ldots, a_n)\) be a sequence of \(\Fm\).
We said that a sequence \(x = (x_1, \ldots, x_n)\) of \(\Fm\) is {\bfseries \(\xi\)-even} if
the sequence \(x(\xi) = (x_1, x_1+a_1, \ldots, x_n, x_n+a_n)\) is even.
We denote by \(E(\xi)\) the number of \(\xi\)-even sequences of \(\Fm\).

\medskip

The last quantity of the previous inequalities is equal to number of \(\xi\)-even sequences times \(\frac{1}{(\theta_1\theta_2)^{2p}}\) where
\(\xi = (a, \ldots, a, b, \ldots, b)\) with \(2p\) times \(a\) and \(2p\) times \(b\).
So we have
\begin{equation}\label{eqP}
  P( \abs{C_a(S_m)} \ge \theta_1 \cap \abs{C_b(S_m)} \ge \theta_2 ) \le \frac{1}{(\theta_1\theta_2)^{2p}}E(\xi)
\end{equation}
and then we will get an upper bound for the number of \(\xi\)-even sequences.

\subsection{Even sequences.}

We give the first properties of \(\xi\)-even sequences.

For all positive integers \(m, n\) such that \(m \le n\), we denote by \(\llb m, n \rrb\) the set of integers strictly between \(m-1\) and \(n+1\).

\begin{lemme}\label{lemme_prop_E}
  Let \(a_1, \ldots, a_n\) be elements of \(\Fm\).
  Then for all element \(c\) of \(\Fm^*\) and for all permutation \(\sigma\) of \(\llb 1, n \rrb\), we have
  \[E(ca_1, \ldots, ca_n) = E(a_{\sigma(1)}, \ldots, a_{\sigma(n)}).\]
\end{lemme}

\begin{proof}
  The map \((x_1, \ldots, x_n) \longmapsto (cx_1, \ldots, cx_n)\) defined a bijection between the set of \((a_1, \ldots, a_n)\)-even sequences
  and the set of \((ca_1, \ldots, ca_n)\)-even sequences.
  Hence, we have
  \[E(ca_1, \ldots, ca_n) = E(a_1, \ldots, a_n)\]
  and it is clear that \(E(a_{\sigma(1)}, \ldots, a_{\sigma(n)}) = E(a_1, \ldots, a_n)\).
\end{proof}

A subset \(J\) of \(\llb 1, n \rrb\) is called a {\bfseries \(\xi\)-subset} if \(\sum_{j \in J} \pm a_j = 0\) for some choice of \(\pm\).

\begin{lemme}\label{lemme_paire}
  If there exists a \(\xi\)-even sequence, then \(\llb 1, n \rrb\) is a \(\xi\)-subset.
\end{lemme}

\begin{proof}
  Let \(x = (x_1, \ldots, x_n)\) be a \(\xi\)-even sequence.
  Thus there exists an element in \(\{ x_1 + a_1, x_2 + a_2, \ldots, x_n + a_n \}\) which must be equal to \(x_1\).
  If \(x_1 = x_1 + a_1\), then \(a_1 = 0\) and \(\sum_{i=1}^n \pm a_i = a_1 + \sum_{i=2}^n \pm a_i = 0\) by induction.
  If \(x_1 = x_2\) and \(y = x_1 + a_1\), then the sequence \((y, y + a_2 - a_1, x_3, x_3+a_3, \ldots, x_n+a_n)\) is even and
  \(\sum_{i=1}^n \pm a_i = (a_2 - a_1) + \sum_{i=3}^n \pm a_i = 0\) by induction.
  If \(x_1 = x_2 + a_2\) and \(y = x_1 + a_1\), then the sequence \((y, y - a_1 - a_2, x_3, x_3+a_3, \ldots, x_n+a_n)\) is even and
  \(\sum_{i=1}^n \pm a_i = -(a_1 + a_2) + \sum_{i=3}^n \pm a_i = 0\) by induction.
\end{proof}

We now give an upper bound on the number of \(\xi\)-even sequences in terms of \(n\) and \(m\).

\begin{lemme}\label{lemme_maj_E_etoile}
  If \(\xi = (a_1, \ldots, a_n)\) is a sequence of \(\Fm\) (where \(n\) is an integer greater than \(1\))
  such that \(a_1 \cdots a_n \not= 0\), then
  \[E(\xi) \le 2^{n-2}(n - 1)!m.\]
\end{lemme}

\begin{proof}
  Let \(x = (x_1, \ldots, x_{n})\) be an \(\xi\)-even sequence.
  Since the sequence \(x(\xi)\) is even, we have \(x_{n} \in \{ x_1, x_1 + a_1, \ldots, x_{n-1}, x_{n-1} + a_{n-1} \}\).
  If \(x_n = x_1\), then the sequence deduced from \(x(\xi)\) by canceling \(x_1\) and \(x_n\) is \(x'(\xi')\)
  where \(\xi' = (a_n - a_1, a_2, \ldots, a_{n-1})\) and \(x' = (x_1+a_1, x_2, \ldots, x_{n-1})\).
  It is cleary even.
  From this, we deduce that the number of \(\xi\)-even sequences \(x = (x_1, \ldots, x_n)\) such that \(x_n = x_1\)
  is lower or equal to \(E(a_n-a_1, a_2, \ldots, a_{n-1})\).
  Similary, we prove that the number of \(\xi\)-even sequences \(x = (x_1, \ldots, x_n)\) such that \(x_n = x_1 + a_1\)
  is lower or equal to \(E(a_{n}+a_1, a_2, \ldots, a_{n-1})\).
  So, we have
  \[E(\xi) \le \sum_{i=2}^n \bigl( E(a_2,  \ldots, a_i + a_1, \ldots, a_n) + E(a_2,  \ldots, a_i - a_1, \ldots, a_n) \bigr).\]
  We have \(E(a_1, a_2) = 0\) if \(a_1 \not= \pm a_2\) and \(m\) if \(a_1 = \pm a_2 \not= 0\).
  By induction, \(E(\xi) \le 2^{n-2}(n - 1)!m\) if \(n \ge 2\).
\end{proof}

We will split \(\xi\)-even sequences into even subsequences and associate to this decomposition a partition of \(\llb 1, n \rrb\).
We said that a sequence \(x = (x_1, \ldots, x_n)\) of \(\Fm\) is {\bfseries exactly \(\xi\)-even} if
the sequence \(x(\xi) = (x_1, x_1+a_1, \ldots, x_n+a_n)\) is even and, for every non-empty proper subset \(J\)
of \(\llb 1, n \rrb\), the sequence \((x_j, x_j+a_j)_{j \in J}\) is not even.

\medskip

We said that a partition of \(\llb 1, n \rrb\) is a {\bfseries \(\xi\)-partition} if each of its blocks
(that is the elements of the partition) is a \(\xi\)-subset.
Given a \(\xi\)-partition \(P = (J_\alpha)_\alpha\) of \(\llb 1, n \rrb\), let \(E(P)\) be the number of sequences
\(x = (x_1, \ldots, x_n)\) of \(\Fm\) such that, for all \(\alpha\),
the sequence \((x_j)_{j \in J_\alpha}\) is \((a_j)_{j \in J_\alpha}\)-even.

\medskip

Let \(x = (x_1, \ldots, x_n)\) be \(\xi\)-even sequence.
If it is not exactly \(\xi\)-even, there exists a non-empty proper \(\xi\)-subset \(J\) of \(\llb 1, n \rrb\) such that
the sequence \((x_j, x_j + a_j)_{j \in J}\) is even.
Since \(x(\xi)\) is even, the sequence \((x_j, x_j + a_j)_{j \in J^c}\) is also even and so \(J^c\) is also a \(\xi\)-subset.
Hence, continuing like this, we obtain a \(\xi\)-partition \((J_\alpha)_\alpha\) of \(\llb 1, n \rrb\)
such that, for all \(\alpha\), the sequence \((x_j, x_j + a_j)_{i \in J_\alpha}\) is even.
So, we have
\[E(\xi) \le \sum_P E(P),\]
where the sum is over the \(\xi\)-partitions \(P\) of \(\llb 1, n \rrb\).

\subsection{Upper bound for $\xi$-partitions}

From now on, unless otherwise stated, we will consider the particular sequence \(\xi = (a_i)_{i \in \llb 1, 4p \rrb}\) of elements of \(\Fm\)
where \(p\) is a positive integer, \(a\) and \(b\) are two coprime integers,
\(a_1 = \cdots = a_{2p} = a\) and \(a_{2p+1} = \cdots = a_{4p} = b\).
We will find an upper bound for \(E(P)\) where \(P\) is a \(\xi\)-partition of \(\llb 1, 4p \rrb\) in terms of its length.
For this, we need a lemma.

\begin{lemme}\label{lemme_maj_produit_factorielles}
  Let \(r\) be a positive integer.
  Let \(N_1, \ldots, N_r\) integers greater than \(2\).
  Then
  \[2^r(N_1 - 1)! \cdots (N_r - 1)! \le 2^2(N_1 + \cdots + N_r - (2r - 1))!.\]
\end{lemme}

\begin{proof}
  Using that, for any integers \(a, b \ge 2\), \(2a!b! \le (a + b - 1)!\) if \(a\) or \(b\) is greater than \(2\),
  we prove the formula by induction.
\end{proof}

Clearly, a \(\xi\)-partition of \(\llb 1, 4p \rrb\) is of length \(2p\) if and only if its blocks are formed of two integers
of \(\llb 1, 2p \rrb\) or of two integers of \(\llb 2p+1, 4p \rrb\).
It follows that for such a partition \(P\), we have \(E(P) = m^{2p}\).

\begin{prop}
  Let \(P\) be a \(\xi\)-partition of \(\llb 1, 4p \rrb\) of length \(2p - k\) where \(k\) is a non-negative integer.
  Then
  \[E(P) \le 2^{2k+2}(2k + 1)!m^{2p - k}.\]
\end{prop}

\begin{proof}
  By the remark preceding the lemma, the inequality is true for \(k = 0\).
  Let \(P = (J_\alpha)_{\alpha=1,\ldots,\ell}\) be a \(\xi\)-partition of length \(\ell = 2p - k\) where \(k\) is a positive integer.
  For all \(\alpha\), let \(N_\alpha\) be the cardinal of \(J_\alpha\).
  Up to renumbering, we can assume that \(N_1 \ge \cdots \ge N_r > N_{r+1} = \cdots = N_\ell = 2\).
  Since \(k \ge 1\), we have \(r \ge 1\).
  We have \(E(P) = \prod_{\alpha=1}^\ell e_\alpha\) where \(e_\alpha = E( (a_j)_{j \in J_\alpha})\).
  By lemmas \ref{lemme_maj_E_etoile} and \ref{lemme_maj_produit_factorielles}, we have
  \[ \prod_{\alpha=1}^r e_\alpha \le m^r 2^{K-r} \prod_{\alpha=1}^r (N_\alpha - 1)! \le 2^{K-2r+2}(K - (2r - 1))!m^r, \]
  where \(K = N_1 + \cdots + N_r\).
  We deduce from \(4p = \sum_{\alpha=1}^\ell N_\alpha = K + 2(\ell - r)\) that \(K = 2k + 2r\).
  So
  \[ \prod_{\alpha=1}^r e_\alpha \le 2^{2k+2}(2k + 1)!m^r. \]
  Since \(e_{r+1} = \cdots = e_\ell = m\), we have \(E(P) \le 2^{2k+2}(2k + 1)!m^\ell\).
\end{proof}

\subsection{Decomposition of \(\xi\)-partitions.}

In this subsection, in the particular case we consider, we will explain how to construct a \(\xi\)-partition of length \(\ell\)
from a \(\xi\)-partition of length \(\ell + 1\).
This will help us in counting the number of \(\xi\)-partitions in the next subsection. 

\medskip

Let \(j, k\) be two non-negative integers.
A \(\xi\)-subset \(J\) of \(\llb 1, 4p \rrb\) is called of {\bfseries type \((j, k)\)} if
the number of elements in \(J \cap \llb 1, 2p \rrb\) (respectively \(J \cap \llb 2p+1, 4p \rrb\)) is \(j\) (respectively \(k\)).

\medskip

From now on, unless otherwise stated, we assume that \(a\) is odd and \(2p(a + b) < m\).

\begin{lemme}\label{lemme_partition_adap}    
  Let \(a, b, p\) and \(\xi\) be as above.
  Let \(J\) be a \(\xi\)-subset of \(\llb 1, 4p \rrb\) of type \((j, k)\).
  \begin{enumerate}
  \item[(a)] If \(j\) and \(k\) are even, then \(J\) is disjoint union of \(\frac{j}{2}\) subsets of type \((2,0)\) and
    of \(\frac{k}{2}\) subsets of type \((0,2)\).
  \item[(b)] The integer \(k\) is odd if and only if \(J\) is disjoint union of one subset of type \((b, a)\) and some subsets of type \((2,0)\)
    and \((0,2)\).
  \end{enumerate}
\end{lemme}

\begin{proof}
  (a) Since \(J\) is of type \((j, k)\), it is the disjoint union of a subset of \(\llb 1, 2p \rrb\) of type \((j, 0)\)
  and of a subset of \(\llb 2p+1, 4p \rrb\) of type \((0, k)\).
  As \(j\) is even, the subset of \(\llb 1, 2p \rrb\) is disjoint union of \(\frac{j}{2}\) subsets of type \((2, 0)\) (because \(a - a = 0\)).
  Similary, the subset of \(\llb 2p+1, 4p \rrb\) is disjoint union of \(\frac{k}{2}\) subsets of type \((0, 2)\).

  \medskip
  
  (b) Assume that \(k\) is odd.
  We can write \(J\) as a disjoint union of a subset of type \((b, a)\) and a subset of type \((j-b, k-a)\).
  So, to prove (b), it suffices to check that the integers \(j - b\) and \(k - a\) are even.
  Since \(J\) is a \(\xi\)-subset, we have \(\sum_{i \in J} \pm a_j = 0\) for some choice of \(\pm\) where \(a_j\) is \(a\) or \(b\).
  Let \(j'\) (respectively \(j''\)) be the number of \(+a\) (respectively of \(-a\))
  and let \(k'\) (respectively \(k''\)) be the number of \(+b\) (respectively of \(-b\)).
  As \(J\) is of type \((j,k)\), we have \(j = j' + j''\) and \(k = k' + k''\).
  On the other hand, we have \((j' - j'')a + (k' - k'')b \equiv 0 \bmod{m}\).
  It follows from \(2p(a + b) < m\) that \(ua = vb \in \Z\) with \(u = j' - j''\) and \(v = k'' - k'\).
  As \(a\) and \(b\) are coprime, we can write \(v = av'\).
  Since the integer \(v'\) is odd, the integers \(j - b = 2j'' -  b(v' - 1)\) and \(k - a = 2k' + a(v' - 1)\) are even.
  Reciprocally, if \(J\) is an disjoint union of one subset of type \((b, a)\), of \(m\) subsets of type \((2,0)\) and
  \(n\) subsets of type \((0,2)\), then \(k = a + 2n\) is odd (since \(a\) is odd).
\end{proof}

Let \(P = (J_\alpha)_{\alpha \in \llb 1, \ell \rrb}\) be a \(\xi\)-partition of \(\llb 1, 4p \rrb\).
Let \(b(P)\) be the number of blocks \(J_\alpha\) such that \(J_\alpha\) is of type \((j_\alpha, k_\alpha)\) where \(k_\alpha\) is odd.
The integer \(2p - b(P)a\) is even since \(2p = \sum_\alpha k_\alpha\) and if \(k_\alpha\) is odd, then \(k_\alpha \equiv a \bmod{2}\) by the previous lemma.
So, since \(a\) is odd, \(b(P)\) is even.

\begin{lemme}\label{lemme_bP_2n}
  Let \(d = a + b - 2\).
  Let \(P\) be a \(\xi\)-partition of \(\llb 1, 4p \rrb\) of length \(\ell\).
  We set \(b(P) = 2n\).
  We have \(\ell \le 2p - nd\), and \(\ell = 2p - nd\) if and only if there is exactly \(2n\) blocks of type \((b, a)\),
  \(p - nb\) blocks of type \((2,0)\) and \(p - na\) blocks of type \((0,2)\) in the partition \(P\).
\end{lemme}

\begin{proof}
  For all \(\alpha\), let  \((j_\alpha, k_\alpha)\) be the type of \(J_\alpha\).
  Up to renumbering the \(J_\alpha\), we can assume that the integers \(k_1, \ldots, k_{2n}\) are odd.
  By lemma \ref{lemme_partition_adap}, the blocks \(J_1, \ldots, J_{2n}\) are of length greater or equal than \(a+b\)
  and the blocks \(J_{2n+1}, \ldots, J_\ell\) are of length greater or equal \(2\).
  Hence, we have
  \[ 4p = \sum_{\alpha=1}^\ell \# J_\alpha \ge 2n(a + b) + 2(\ell - 2n) = 2nd + 2\ell, \]
  so \(\ell \le 2p - nd\).
  We have \(\ell = 2p - nd\) if and only if \(\# J_1 = \cdots = \# J_{2n} = a + b\) and \(\# J_{2n+1} = \cdots = \# J_\ell = 2\).
\end{proof}

\begin{prop}\label{lemme_part_reunion_blocs}
  Let \(P\) be a \(\xi\)-partition of \(\llb 1, 4p \rrb\) of length \(\ell\) with \(b(P) = 2n\).
  If \(\ell < 2p - nd\), then there exists a \(\xi\)-partition \(P' = (J'_\beta)_{\beta \in \llb 1, \ell+1 \rrb}\) such that \(b(P) = b(P')\) and
  \[ P = (J'_1 \cup J'_2, J'_3, \ldots, J'_{\ell+1}) \]
  up to a permutation.
\end{prop}

\begin{proof}\
  Assume that \(n > 0\).
  Up to renumbering the \(J_\alpha\), we can assume that \(J_1\) of type \((j, k)\) with \(k\) odd.
  By lemma \ref{lemme_partition_adap}, we can write \(J_1 = J'_1 \cup J''_1\) where  \(J'_1\) is a \(\xi\)-subset of type \((b, a)\)
  and \(J'_2\) is a \(\xi\)-subset of type \((j - b, k - a)\) with \(k - a\) even.%with \(j - b\) and \(k - a\) even.
  
  \medskip
  
  Assume that \(n = 0\).
  Then the type of \(J_1\) is \((j, k)\) with \(j\) and \(k\) even.
  If \(j > 0\), then \(J_1 = J'_1 \cup J''_1\) where \(J'_1\) is a \(\xi\)-subset of type \((2, 0)\) and \(J'_2\) is a \(\xi\)-subset of type \((j - 2, k)\).
\end{proof}

\subsection{Bound on the number of \(\xi\)-partitions.}

We still assume \(\xi = (a_1, \ldots, a_{4p})\) is the sequence of \(\Fm\) where \(a_1 = \cdots = a_{2p} = a\)
and \(a_{2p+1} = \cdots = \rg{a_{4p} =b}\) with \(a\) and \(b\) two coprime integers such that \(a\) \rg{is} odd.
We will now find an upper bound for the number of \(\xi\)-partitions \(P\) of \(\llb 1, 4p \rrb\) in terms of its length and of the integer \(b(P)\)
which has been defined just after the proof of lemma 10.

\medskip

We assume that \(2p(a + b) < m\) and we let \(d = a + b - 2\).
For all non-negative integers \(n\) and \(k\), let \(c_k^{(n)}\) be the number of \(\xi\)-partitions \(P\) of \(\llb 1, 4p \rrb\) of
length \(2p - k\) such that \(b(P) = 2n\) and let
\[ C_k^{(n)} = \sum_{R} E(R), \]
where the sum is over the \(\xi\)-partitions \(R\) of length \(2p - k\) such that \(b(R) = 2n\).
We have \(c_k^{(n)} = 0 = C_k^{(n)}\) if \(k < nd\) or \(n > N\) where \(N = \min(\floor{p/a}, \floor{p/b})\).
We also have \(c_0^{(0)} = (2p - 1)!!^2\) and \(C_0^{(0)} = (2p - 1)!!^2m^{2p}\) where
\[(2p - 1)!! = \frac{(2p)!}{p!2^p} = (2p - 1)(2p - 3) \cdots 3 \cdot 1\]
is the double factorial \(2p - 1\), the number of ways to arrange \(2p\) objects into \(p\) unordered pairs.

\medskip

We have therefore
\begin{align}\label{Exi}
  E(\xi) &\le \sum_{e=0}^N \sum_{k=ed}^{2p-1} C_k^{(e)} \\
         &= (2p - 1)!!^2m^{2p} + \sum_{k=1}^{d-1} C_k^{(0)} + \sum_{e=1}^{N-1}\sum_{k=ed}^{(e+1)d-1} \sum_{n=0}^e C_k^{(n)}
           + \sum_{k=Nd}^{2p-1} \sum_{n=0}^N C_k^{(n)}. \nonumber
\end{align}
So, to get an upper bound for \(E(\xi)\), we will study \(\sum_{n=0}^e C_k^{(n)}\).
Since, by lemma \ref{lemme_maj_E_etoile}, we have \(C_k^{(n)} \le c_k^{(n)}2^{2k+2}(2k + 1)!m^{2p-k}\),
it suffices to study \(c_k^{(n)}\).

\begin{lemme}\label{lemme_maj_ck}
  Let \(a\), \(b\), \(p\) and \(d\) be as above.
  \begin{enumerate}
  \item[(a)] For all integer \(k \in \llb 1, 2p-1 \rrb\), we have
    \[c_k^{(0)} \le 2^{k-2}(2p)(2p - 1)p^{2k-2}(2p -1)!!^2.\]
  \item[(b)] If \(n\) is a positive integer such that \(n \le N\), then
    \[c_{nd}^{(n)} \le 2^n p(p - 1)p^{n(d+2)-2}(2p - 1)!!^2.\]
  \item[(c)] For all non-negative integers \(n\) and \(j\), we have
    \[ c_{nd+j}^{(n)} \le 2^{n+j}p^{n(d+2)+2j}(2p - 1)!!^2.\]
  \end{enumerate}
\end{lemme}

\begin{proof}
  (a) By the proposition \ref{lemme_part_reunion_blocs}, we have \(c_1^{(0)} \le \frac{1}{2}(2p)(2p - 1)(2p - 1)!!^2\) and, by induction,
  \[c_{k+1}^{(0)} \le \mybinom{2p - k}{2}c_k^{(0)} \le 2^{k-1}(2p)(2p - 1)p^{2k}(2p -1)!!^2.\]

  (b) By lemma \ref{lemme_bP_2n}, a \(\xi\)-partition \(P\) of length \(2p - nd\) such that \(b(P) = 2n\)
  consists of exactly \(2n\) blocks of type \((b, a)\), \(p - nb\) blocks of type \((2,0)\) and \(p - na\) blocks of type \((0,2)\).
  Hence, we have
  \begin{align*}
    c_{nd}^{(n)} &= \prod_{i=0}^{2n-1} \mybinom{2p - ia}{a}\mybinom{2p - ib}{b}(2p - 2nb - 1)!!(2p - 2na - 1)!!	\\
                 &= \Bigl( \frac{2^{a+b}}{a!^2b!^2} \Bigr)^n \prod_{j=0}^{an-1} (p - j) \prod_{k=0}^{bn-1} (p - k) \cdot (2p - 1)!!^2 \\
                 &\le 2^n p(p - 1)p^{n(a + b)-2}(2p - 1))!!^2.
  \end{align*}
  
  (c) By (b), the formula is true for \(j = 0\).
  By the proposition \ref{lemme_part_reunion_blocs}, we have \(c_{nd+j+1}^{(n)} \le \mybinom{2p - nd - j}{2} c_{nd+j}^{(n)}.\)
  So, by induction, we obtain
  \begin{align*}
    \frac{1}{(2p - 1)!!^2} c_{nd+j+1}^{(n)} &\le \frac{(2p - nd - j)(2p - nd - j - 1)}{2} 2^{n+j}p^{n(d + 2)+2j} \\
                                            &\le 2^{n+j+1}p^{n(d+2)+2j+2}. \qedhere
  \end{align*}
\end{proof}

We will now find an upper bound for the sum of \(c_k^{(n)}\).

\begin{lemme}\label{lemme_maj_sum_cnk}
  Let \(a\), \(b\), \(p\) and \(d\) be as above.
  For all \(e \in \llb 0, N \rrb\) and all \(k \in \llb ed, 2p - 1 \rrb\), we have
  \[ \sum_{n=0}^e c_k^{(n)} \le 2^{k+1} p^{3k}(2p - 1)!!^2. \]
\end{lemme}

\begin{proof}
  The lemma is trivial if \(e = k = 0\) and it follows from lemma \ref{lemme_maj_ck}(a) if \(e = 0\) and \(k \ge 1\).
  Assume that \(e \ge 1\).
  We first consider the case where \(d > 1\).
  For all \(k \in \llb ed, 2p - 1 \rrb\) and all \(n \in \llb 1, e \rrb\), we have
  \[ \frac{1}{(2p - 1)!!^2}c_k^{(n)} \le 2^{ed+j}p^{2ed+2j} (2^{d-1}p^{d-2})^{-n} \]
  by lemma \ref{lemme_maj_ck}(c).
  Hence, we have
  \begin{align*}
    \frac{1}{(2p - 1)!!^2}\sum_{n=0}^e c_k^{(n)} \le 2^{k-2}(2p)(2p - 1)p^{2k-2} + 2^{e+j}p^{(d+2)e+2j}
    \frac{2^{e(d-1)}p^{e(d-2)} - 1}{2^{d-1}p^{d-2} - 1}.
  \end{align*}
  To prove the lemma in this case, it suffices to prove that the right hand side is lower or equal to \(2^{k+1}p^{2k} - 2^{k-1}p^{2k-1}\).
  It follows from the trivial inequality \(2^{d-1}p^{d-2} \ge 2\).

  \medskip

  Assume now that \(d = 1\).
  Let  \(k \in \llb e, 2p-1 \rrb\).
  By lemmas \ref{lemme_maj_ck} (b) and (c), we have
  \begin{align*}
    \frac{1}{(2p - 1)!!^2} \sum_{n=0}^e c_k^{(n)} &\le \sum_{n=0}^{k-1} 2^{n+(k-n)}p^{3n+2(k-n)} + 2^k p(p - 1)p^{3k-2} \\
                                                  &\le 2^k p^{2k} \frac{p^{k} - 1}{p-1} + 2^{k+1} p(p - 1)p^{3k-2}.
  \end{align*}
  We check that the right hand side is lower or equal to \(2^{k+1}p^{3k}\).
\end{proof}

\begin{prop}\label{lemme_maj_Exi}
  For all \(m\) sufficiently large and all distinct elements \(a\) and \(b\) of \(\Fm\) such that \(a + b < \frac{m}{2\log m}\),
  we have
  \[ E(\xi) \le 2(2p-1)!!^2 m^{2p} \]
  where \(p = \floor{\log m}\) and \(\xi\) is the sequence \((a, \ldots, a, b, \ldots, b)\) of \(\Fm\) of length \(4p\)
  with \(2p\) times \(a\).
\end{prop}

\begin{proof}
  By lemma \ref{lemme_prop_E}, we can assume that \(a\) and \(b\) are coprime integers and that \(a\) is odd.
  By (\ref{Exi}), we have
  \[
    E(\xi) \le (2p - 1)!!^2m^{2p} + \sum_{k=1}^{d-1} C_k^{(0)} + \sum_{e=1}^{N-1}\sum_{k=ed}^{(e+1)d-1} \sum_{n=0}^e C_k^{(n)}
    + \sum_{k=Nd}^{2p-1} \sum_{n=0}^N C_k^{(n)}.
  \]
  By lemmas \ref{lemme_maj_E_etoile} and \ref{lemme_maj_sum_cnk}, we have
  \[
    \sum_{n=0}^e C_k^{(n)} \le 2^{2k+2}(2k + 1)! m^{2p-k} \sum_{n=0}^e c_k^{(n)} \le 2^{7k+5} p^{5k+1}m^{2p-k}(2p - 1)!!^2.
  \]
  As the right hand is lower or equal to
  \[\Bigl(1 + 2^{12}p^6m^{-1} + 2^{19}p(2p - 2)(p^5m^{-1})^2 \Bigr) (2p-1)!!^2 m^{2p}\]
  for all integer \(m\) big enough, we can conclude.
\end{proof}

\subsection{Proof of proposition \ref{prop_P(Cu_cap_Cv)}.}

We can now prove the proposition \ref{prop_P(Cu_cap_Cv)}.
Using (\ref{eqP}) and lemma \ref{lemme_maj_Exi}, for all \(m\) sufficiently large and all elements \(u\) and \(v\) of \(\Fm^*\)
  such that \(u\) and \(v\) are coprime and \(u + v < \frac{m}{2\log m}\), we have
\[
  P( \abs{C_u(S_m)} \ge \theta_1 \cap \abs{C_v(S_m)} \ge \theta_2 )
  \le \frac{1}{(\theta_1\theta_2)^{2p}}E(\xi)
  \le \frac{2}{(\theta_1\theta_2)^{2p}}(2p-1)!!^2 m^{2p}
\]
where \(p = \floor{\log m}\) and \(\xi = (a_1, \ldots,  a_{4p})\) is the sequence of \(\F_m\)
such that \(a_1 = \cdots = a_{2p}=u\) and \( a_{2p+1} = \cdots = a_{4p} = v\).
If we take \(\theta_1 = \theta_2 = \lambda_m\), then we have
\[
  P( \abs{C_u(S_m)} \ge \theta_1 \cap \abs{C_v(S_m)} \ge \theta_2) \le 2\frac{(2p - 1)!!^2}{2^{2p}\log^{2p} m}.
\]
We deduce from Stirling's approximation that \(\frac{(2p - 1)!!^2}{2^{2p}\log^{2p} m} \le \frac{3e^2}{m^2}\).
So, for all \(m\) sufficiently large, we have \(P( \abs{C_u(S_m)} \ge \theta_1 \cap \abs{C_v(S_m)} \ge \theta_2) \le \frac{6e^2}{m^2}\).

\bibliographystyle{amsplain}

\end{document}